\documentclass{article}
\usepackage{arxiv}
\usepackage[T1]{fontenc}
\usepackage[utf8]{inputenc}

\usepackage{amsmath,amsthm}
\usepackage{graphicx}
\usepackage{diagbox}
\usepackage{algorithm}
\usepackage{algpseudocode}
\usepackage{url}

\usepackage{hyperref}
\usepackage{xcolor}
\hypersetup{
	colorlinks,
	linkcolor={black!50!black},
	citecolor={blue!50!blue},
	urlcolor={blue!80!black}
}

\theoremstyle{plain}
\newtheorem{theorem}{Theorem}[section]

\newtheorem{proposition}[theorem]{Proposition}

\theoremstyle{definition}

\theoremstyle{remark}

\usepackage{amsfonts}
\usepackage{mathtools}

\DeclareMathOperator*{\E}{\mathbb{E}}

\begin{document}
	
	\title{Coloured Tobit Kalman Filter}
	
	\author{Kostas Loumponias
		\thanks{K. Loumponias is with the Department of Mathematics, Aristotle University of Thessaloniki, GR-54124, Thessaloniki, Greece (e-mail: kostikasl@math.auth.gr) }
	}
	
	\maketitle
	
	\begin{abstract}
		
		This paper deals with the Tobit Kalman filtering (TKF) process when the one- dimensional measurements are censored and the noises of the state-space model are coloured. Two improvements of the standard TKF process are proposed. Firstly, the exact moments of the censored measurements are calculated via the moment generating function of the censored measurements. Secondly, coloured noises are considered in the proposed method in order to tackle real-life problems, where the white noises are not common. The designed process is evaluated using two experiments-simulations. The results show that the proposed method outperforms other methods in minimizing the Root Mean Square Error (RMSE) in both experiments.

	\end{abstract}

	\textit{Keywords}: Censored Data, Moment Generating Function, Kalman Filtering, Coloured Noise.

	\section{Introduction}
	
	The well-known Kalman filter ( KF ) \cite{brown1992introduction} process is a recursive linear filter, which  provides estimates of the hidden state vectors by using a series of measurements observed over time. The KF process provides an optimal performance (i.e., unbiased estimations with minimum variance errors) when i) the  state and measurement equations are linear and ii) the measurement and process noises  are white (usually normally distributed). However in many real-life problems the state-space model is not linear, therefore,  KF  is not suitable for providing optimal estimates. The most known filters, that  have been proposed in order to overcome the problem of non-linearity, are the Extended Kalman filter (EKF) \cite{ljung1979asymptotic}, the Unscented Kalman Filter ( UKF ) \cite{wan2000unscented} and the Particle filter ( PF ) \cite{djuric2003particle}.

	In the case where  non-linearity is  due to censoring in the measurements \cite{dong2004food}, the UKF and EKF  result in calculating biased estimates for the hidden state vectors, as it has been shown in \cite{allik2014tobit,allik2014estimation}. PF  is able to cope with censored measurements, however, it imposes a heavy computational burden. To that end, Tobit Kalman filter ( TKF ) \cite{allik2015tobit} has been proposed to cope with censored measurements with a low computational cost. More specifically, the  TKF process provides unbiased estimates when dealing with censored measurements using the Tobit model of type I with two censoring limits \cite{loumponias2016using,loumponias2016usingb}.  In \cite{loumponias2018adaptive}, a multi-object tracking algorithm based on the TKF process was proposed, where the  exact variance of the censored measurement is calculated, while in \cite{allik2015tobit}  an approximated censored variance is used. In \cite{loumponias2018adaptive}, a brief proof for the calculation of the exact  censored variance is provided. Other variants of TKF dealing with fading measurements and correlated noises are proposed in \cite{geng2017tobit} and \cite{li2016tobit}, respectively.

	In all the aforementioned methods, the measurement and process noises are assumed to be white. Nevertheless, in many real problems \cite{kuhlmann2003kalman}, this assumption fails and the noises can be adequately described by $AR(p)$ models \cite{akaike1998autoregressive}; such noises are called “coloured” noises. In the case where the noises are coloured, the standard KF,  EKF ,  UKF etc. provide biased estimations, since, they cope with white noises. In order to overcome this drawback, the state-space model with coloured noises is written in the form of a system driven by white noises. To that end, two non-numerical methods are proposed, the augmented \cite{bryson1968estimation,popescu1998kalman}  and the measurement differencing \cite{anderson2012optimal,chen1986design} approach. 
	
	In the augmented approach, the  coloured noise of the measurement is included into the state vector. By doing so, the measurements of the augmented system are "perfect", i.e., they do not longer contain  noise. Hence, the covariance matrix of the measurement at time $ t $ given the a priori estimations of the state vector up to time $ t-1 $ may become ill-conditioned, i.e., a singular matrix. In the measurement differencing approach, the derived measurements are expressed in terms of the state vector by considering white noise. In effect, a linear combination of two measurements in sequence is determined in order to eliminate the coloured noise. However, the differencing approach leads to a risk of unstable solution when inaccurate observations occur \cite{gazit1997digital}.

	The main contribution of this paper is the establishment of the coloured TKF ( ColTKF) dealing with censored measurements when the process and measurement noises are coloured. In accordance with other studies dealing with censored measurements  \cite{allik2014estimation,allik2015tobit,li2016tobit,loumponias2016using}, the exact censored moments (order one to three) are calculated and also their properties are described. Furthermore, this paper deals with a) multi-dimensional hidden state vector, b) one-dimensional censored measurement (Tobit Type I) $\cite{wooldridge2002econometric}$ and c) coloured noises described by the $AR(1)$ model. To that end, the moment generating function (mgf) of the multivariate Gaussian distribution is calculated when a marginal variable is censored with two censoring limits. Then, the marginal mgfs of the censored and uncensored variables are derived. Therefore, censored moments -of order one to three- can be calculated and it is proved that the rest variables (the uncensored) are still normally distributed. Next,  ColTKF is derived by using the augmented approach and the moments of the censored measurements. Finally, the likelihood function of the censored measurements is utilised in order to estimate the unknown parameters of AR(1) models in coloured noises. 
	
	In this paper, the augmented approach is preferred, since, the censored measurements are one-dimensional, therefore, the computational burden is only slightly increased. Furthermore, in the measurement differencing approach, two latent measurements (not censored) in sequence are utilized in order the derived measurements to be expressed in terms of white noises; nevertheless, in the case of censoring, the latent measurements are not provided. The results in the simulations show that ColTKF has a better performance than TKF and the augmented KF, as it was expected. 
	
	The rest of the paper is organised as follows: In Section 2, the moments of a censored one-dimensional variable are calculated by means of the mgf. In Section 3, the proposed method (ColTKF) is provided. In Section 4, experimental results are illustrated using artificial data to show up  the effectiveness of the proposed process. Finally, in Section 5, conclusions are provided.

	\section{Moment Generating Function: Censored Case}
	Let $\textbf{x}\sim N(\textbf{m}, \textbf{S})$ and $\textbf{x}_{-k}=(x_1,...,x_{k-1},x_{k+1},...,x_{n})^T$, where $\textbf{m}=\{m_i\}_{i=1}^n$ and $\textbf{S}=\{S_{i,j}\}_{i,j=1}^n$. We consider the case where the $k^{th}$ random variable (rv) of $\textbf{x}$ is censored, symbolized by $x^c_{k}$.  More specifically, let
	\begin{equation}
	x^c_{k}=
	\begin{cases} 
	x_k,&a<x_k<b\\
	a,&x_k \leq a\\
	b,&x_k \geq b
	\end{cases}.
	\label{censored}
	\end{equation}
	where $a$ and $b$ are the censoring limits. 
	
	Then, the distribution of the random vector $\textbf{x}^c=(x_1,...,x_{k-1},x_k^c,x_{k+1},...,x_{n})^T$ is given by 
	\begin{align}
	\nonumber
	f_{\textbf{x}^c}(\textbf{x}^c) =& f_{\textbf{x}}(\textbf{x}^c)u_{(a,b)}(x_k) \\
	\nonumber			
	&+ \int_{-\infty}^{a} f_{\textbf{x}}(\textbf{x}^c)dx_k\delta(x_k-a) \\
	&+ \int_{b}^{+\infty} f_{\textbf{x}}(\textbf{x}^c)dx_k\delta(x_k-b), \label{dist}
	\end{align}
	where $f_\textbf{x}(\textbf{x})$ is the probability distribution function (pdf) of \textbf{x} and $\delta$ stands for the Kronecker delta function. Then,
	the  moment generating function (mgf) of the censored random vector $\textbf{x}^c$ can be derived by using (\ref{dist}): 
	
	\begin{proposition}
		For a normally distributed random vector $\textbf{x}\sim N(\textbf{m}, \textbf{S})$, the mgf of $\textbf{x}^c$ is given by 
		\begin{align*}
		M_{\textbf{x}^c}(\textbf{t}) &= exp(\textbf{t}^T\textbf{m} + \frac{1}{2}\textbf{t}^t\textbf{S}\textbf{t} )(F_{e_k}(b-m_k-\sum\limits_{i=1}^nS_{k,i}t_{i})-F_{e_k}(a-m_k-\sum\limits_{i=1}^nS_{k,i}t_{i}))\\
		&+exp(t_k a + \textbf{t}_0^T\textbf{m} + \textbf{t}_0^T\textbf{S}\textbf{t}_0/2)F_{e_k}(a-m_k-\sum\limits_{i=1}^nS_{k,i}t_{0,i})\\
		&+exp(t_k b + \textbf{t}_0^T\textbf{m}+ \textbf{t}_0^T\textbf{S}\textbf{t}_0/2)(1-F_{e_k}(b-m_k-\sum\limits_{i=1}^nS_{k,i}t_{0,i})),
		\end{align*}
		where $\textbf{t}_0=(t_1,...,t_{k-1},0,t_{k+1},...,t_n)^T$, $e_k\sim N(0,S_{k,k}$) and $F_{e_k}$ stands for the cumulative distribution of $e_k$.
		\label{p1}
	\end{proposition}
	
	\begin{proof}
		We have that\\
		\begin{equation}
		\begin{aligned} 
		M_{\textbf{x}^c}(\textbf{t}) = Ee^{\textbf{t}^T\textbf{x}} =&
		\int_{R^{n-1}}\int_a^b e^{\textbf{t}^T\textbf{x}}f_{\textbf{x}}(\textbf{x})d\textbf{x}\\
		&+\int_{R^{n-1}}e^{\textbf{t}_{-k}^T\textbf{x}_{-k}+t_ka}\Big(\int_{-\infty}^a f_{\textbf{x}}(\textbf{x})dx_k\Big)d\textbf{x}_{-k}\\
		&+\int_{R^{n-1}}e^{c^T\textbf{x}_{-k}+t_kb}\Big(\int^{\infty}_b f_{\textbf{x}}(\textbf{x})dx_k\Big)d\textbf{x}_{-k},
		\end{aligned}
		\label{moment}
		\end{equation}
		where $\textbf{t}_{-k}=(t_1,...,t_{k-1},t_{k+1},...,t_n)^T$. The first term on the right hand side of (\ref{moment}) reads
		\begin{equation*}
		A_1(\textbf{t}) = \int_{R^{n-1}}\int_a^b e^{\textbf{t}^T\textbf{x}}
		\frac{1}{2\pi|\textbf{S}|^{1/2}}exp(-\frac{1}{2}(\textbf{x}-\textbf{m})^T\textbf{S}^{-1}(\textbf{x}-\textbf{m}))d\textbf{x} ,
		\label{A1_1}
		\end{equation*}
		which for $\textbf{x}^* = \textbf{x} -\textbf{m}$  yields
		\begin{align}
		\nonumber
		A_1(\textbf{t}) &= \int_{R^{n-1}}\int_{a-m_k}^{b-m_k} e^{\textbf{t}^T(\textbf{x}^*+\textbf{m})}
		\frac{1}{2\pi|\textbf{S}|^{1/2}}exp(-\frac{1}{2}\textbf{x}^{*T}\textbf{S}^{-1}\textbf{x}^*)d\textbf{x}^*\\
		\nonumber
		&=\frac{e^{\textbf{t}^T\textbf{m}}}{2\pi|\textbf{S}|^{1/2}}
		\int_{R^{n-1}}\int_{a-m_k}^{b-m_k}exp(-\frac{1}{2}\textbf{x}^{*T}\textbf{S}^{-1}\textbf{x}^* + \textbf{t}^T\textbf{x}^* )d\textbf{x}^*.
		\end{align}
		Then, for $\textbf{j}=\textbf{S}\textbf{t}$  we get
		\begin{align}
		\nonumber
		A_1(\textbf{t}) &=\frac{e^{\textbf{t}^T\textbf{m}}}{2\pi|\textbf{S}|^{1/2}}
		\int_{R^{n-1}}\int_{a-m_k}^{b-m_k}exp(-\frac{1}{2}(\textbf{x}^{*}-\textbf{j})^T\textbf{S}^{-1}(\textbf{x}^*-\textbf{j}) + \frac{1}{2}\textbf{t}^T\textbf{S}\textbf{t} )d\textbf{x}^*\\
		&=\frac{exp(\textbf{t}^T\textbf{m} + \frac{1}{2}\textbf{t}^t\textbf{S}\textbf{t} )}{2\pi|\textbf{S}|^{1/2}}\int_{R^{n-1}}\int_{a-m_k}^{b-m_k}exp(-\frac{1}{2}(\textbf{x}^{*}-\textbf{j})^T\textbf{S}^{-1}(\textbf{x}^*-\textbf{j}))d\textbf{x}^*.
		\label{A1_3}
		\end{align}
		Now, for $\textbf{e} = \textbf{x}^* - \textbf{j}$, (\ref{A1_3}) leads to
		\begin{align}
		\nonumber
		A_1(\textbf{t})&=\frac{exp(\textbf{t}^T\textbf{m} + \frac{1}{2}\textbf{t}^t\textbf{S}\textbf{t} )}{2\pi|\textbf{S}|^{1/2}}\int_{R^{n-1}}\int_{a-m_k-j_k}^{b-m_k-j_k}exp(-\frac{1}{2}\textbf{e}^T\textbf{S}^{-1}\textbf{e})d\textbf{e}\\
		&=exp(\textbf{t}^T\textbf{m} + \frac{1}{2}\textbf{t}^t\textbf{S}\textbf{t} )(F_{e_k}(b-m_k-j_k)-F_{e_k}(a-m_k-j_k)),
		\label{A1_4}   
		\end{align}
		where $F_{e_k}(e)$ is the marginal cumulative function of the random variable $e_k\sim N(0,S_{k,k})$, and $j_k = \sum\limits_{i=1}^nS_{k,i}t_{i}$.
		
		Next, the second term of (\ref{moment}) is computed as follows:
		
		\begin{align}
		\nonumber
		A_2(\textbf{t}) &= \frac{e^{t_k a}}  {2\pi|\textbf{S}|^{1/2}}\int\limits_{R^{n-1}}\int\limits_{-\infty}^a
		exp(-\frac{1}{2}(\textbf{x}-\textbf{m})^T\textbf{S}^{-1}(\textbf{x}-\textbf{m})+\textbf{t}^T_{-k}\textbf{x}_{-k})d\textbf{x}\\
		&=\frac{exp(t_k a + \textbf{t}^T_{-k}\textbf{m}_{-k})}{2\pi|\textbf{S}|^{1/2}}\int\limits_{R^{n-1}}\int\limits_{-\infty}^{a-m_k}
		exp(-\frac{1}{2}\textbf{x}^{*T}\textbf{S}^{-1}\textbf{x}^*+\textbf{t}^T_{-k}\textbf{x}_{-k})d\textbf{x}^{*}.
		\label{A2_1}  
		\end{align}
		Then, for $\textbf{j}_0=\textbf{S}\cdot\textbf{t}_0$, (\ref{A2_1}) becomes
		
		\begin{align}
		\nonumber
		A_2(\textbf{t}) &= \frac{exp(t_k a + \textbf{t}_0^T\textbf{m} + \textbf{t}_0^T\textbf{S}\textbf{t}_0/2)}{2\pi|\textbf{S}|^{1/2}} \kern-0.8em
		\int\limits_{R^{n-1}} \kern-0.4em \int\limits_{-\infty}^{a-m_k}  \kern-0.8em
		exp(-\frac{1}{2}(\textbf{x}^*-\textbf{j}_0)^{*T}\textbf{S}^{-1}(\textbf{x}^*-\textbf{j}_0))d\textbf{x}^{*}\\
		\nonumber
		&=\frac{exp(t_k a + \textbf{t}_0^T\textbf{m}+ \textbf{t}_0^T\textbf{S}\textbf{t}_0/2)}{2\pi|\textbf{S}|^{1/2}} \kern-0.8em
		\int\limits_{R^{n-1}} \kern-0.4em \int\limits_{-\infty}^{a-m_k-j_{0,k}}  \kern-0.8em
		exp(-\frac{1}{2}\textbf{e}^T\textbf{S}^{-1}\textbf{e})d\textbf{e}\\
		&=exp(t_k a + \textbf{t}_0^T\textbf{m} + \textbf{t}_0^T\textbf{S}\textbf{t}_0/2)F_{e_k}(a-m_k-j_{0,k}),
		\label{A2_2}
		\end{align}
		where $j_{0,k} = \sum\limits_{i=1}^nS_{k,i}t_{0,i}$.\\
		In the same way as for the second term, the third term of  (\ref{moment}) is given by
		\begin{equation}
		A_3(\textbf{t})=exp(t_k b + \textbf{t}_0^T\textbf{m} + \textbf{t}_0^T\textbf{S}\textbf{t}_0/2)(1-F_{e_k}(b-m_k-j_{0,k})).
		\label{A3_1}
		\end{equation}
		Finally, we get the analytic form of the MGF  by  substituting (\ref{A1_4}), (\ref{A2_2}) and (\ref{A3_1}) into (\ref{moment}). 
	\end{proof}
	
	It is derived from Proposition \ref{p1}, that the mgf of the marginal $\textbf{x}_{-k}$ is equal with \\
	\begin{equation}
	M_{\textbf{x}_{-k}}(\textbf{t}) = M_{\textbf{x}^c}(\textbf{t}_0) = exp(\textbf{t}_0^T\textbf{m} + \frac{1}{2}\textbf{t}_0^t\textbf{S}\textbf{t}_0 ),
	\label{mrg1}
	\end{equation}
	thus, $\textbf{x}_{-k}\sim N(\textbf{m}_{-k},\textbf{S}_{-k,-k})$. We notice that this result does not hold in the case where the random variable $x_k$ is truncated \cite{arismendi2013multivariate}.
	
	In the same way, the mgf of the censored variable $x^c_k$ is given by 
	
	\begin{align}
	\nonumber
	M_{x^c_k}(t_k) &= exp(t_km_k + \frac{1}{2}t^2_kS_{k,k})(F_{e_k}(b-m_k-S_{k,k}t_k)-F_{e_k}(a-m_k-S_{k,k}t_k))\\
	&+exp(t_k a)F_{e_k}(a-m_k) + exp(t_k b)(1-F_{e_k}(b-m_k)),
	\label{mrg2}
	\end{align}
	which has the same form as in \cite{loumponias2018adaptive}. 
	
	Next, the censored mean, variance and skewness of $x^c_k$,  can be calculated by (\ref{mrg2}):
	\begin{align}
	\nonumber
	\E(x^c_k) &= \frac{dM_{x^c_k}(t)}{dt}\Big|_{t=0} = aF_{e_k}(a-m_k) + b(1-F_{e_k}(b-m_k))\\
	&+(F_{e_k}(b-m_k)-F_{e_k}(a-m_k))m_k + S_{k,k}(f_{e_k}(a-m_k)-f_{e_k}(b-m_k)),
	\label{cen_mean}
	\end{align}
	\begin{align}
	\nonumber
	Var(x^c_k) &= \frac{d^2M_{x^c_k}(t)}{dt^2}\Big|_{t=0} - \big(\E(x^c_k)\big)^{2}=\\
	\nonumber
	&=a^2F_{e_k}(a-m_k)(1-F_{e_k}(a-m_k))+b^2F_{e_k}(b-m_k)(1-F_{e_k}(b-m_k))\\
	\nonumber
	&+m_kP(1-P) + S^2_{k,k}P +  2m_kS_{k,k}(f_{e_k}(a-m_k)-f_{e_k}(b-m_k))\\
	\nonumber
	&+S^2_{k,k}((a-m_k)f_{e_k}(a-m_k)-(b-m_k)f_{e_k}(b-m_k))\\
	\nonumber
	&-2abF_{e_k}(a-m_k)(1-F_{e_k}(b-m_k))\\ 
	\nonumber
	&-S^2_{k,k}(f_{e_k}(a-m_k)-f_{e_k}(b-m_k))^2\\
	\nonumber
	&-2\big[Pm_k + S_{k,k}(f_{e_k}(a-m_k)-f_{e_k}(b-m_k))\big]\cdot\\
	&\cdot\big[aF_{e_k}(a-m_k) + b(1-F_{e_k}(b-m_k))\big],
	\label{cen_var}
	\end{align}
	where $P=F_{e_k}(b-m_k)-F_{e_k}(a-m_k)$. 
	
	The third moment of $x^c_k$ is given by 
	\begin{align}
	\nonumber
	\E\big(x^{c3}_k\big) &=\frac{d^3M_{x^c_k}(t)}{dt^3}\Big|_{t=0} =\\ 
	\nonumber
	&=(m_k^3 + 3m_ks^2)P + a^3F_{e_k}(a-m_k) + b^3(1-F_{e_k}(b-m_k))\\
	\nonumber
	&+(2s^3 + 3m_k^2s)(\phi(a^*)-\phi(b^*)) + 3m_ks^2(a^*\phi(a^*)-b^*\phi(b^*))\\
	&+s^3(a^{*2}\phi(a^*)-b^{*2}\phi(b^*)),\label{cen_mom3}
	\end{align}
	where $s= \sqrt{S_{k,k}}$, $\phi(x)$ stands for the   probability density function of the standard normal distribution, 
	$a^* = (a-m_k)/s$ and $b^* = (b-m_k)/s$. Then, the coefficient of the censored skewness, $\gamma^c_k$, \cite{azzalini2013skew} is calculated by substituting (\ref{cen_mean})-(\ref{cen_mom3}) into 
	\begin{equation}
	\gamma^c_k = \frac{\E\big(x^{c3}_k\big)- 3\E\big(x^{c}_k\big)Var\big(x^c_k\big) -\E\big(x^{c}_k\big)^3}{Var(x^c_k)^{(3/2)}}.
	\label{cen_skw}
	\end{equation}
	
	Furthermore, the covariance of the variables $x_i$ and $x^c_k$, for $i \neq k$, can be calculated by means of Proposition \ref{p1} and (\ref{cen_mean}):
	
	\begin{align}
	\nonumber
	Cov(x_i,x_k^c)&=\E(x_ix_k^c) - \E(x_i)\E(x_k^c)\\  
	\nonumber
	&=\frac{dM_{x_i,x_k^c}(t_i,t_k)}{dt_idt_k}\Big|_{(t_i,t_j)=(0,0)} - m_i\E(x_k^c)\\
	&=P\cdot S_{i,k},
	\label{cen_cov}
	\end{align}
	where $S_{i,k}=cov(x_i,x_k)$. Furthermore, it is derived by (\ref{cen_mean}) and (\ref{cen_skw}) that 
	\[m^c_k = m_k\]
	and
	\[\gamma^c_k = 0,\]
	for the censoring limits $a$ and $b$, which fulfill the conditions, $a < m_k$  and $ m_k = (a+b)/2$. 
	
	$\textbf{An illustrative example}$. In order to verify the aforementioned results, let us consider the censored mean vector  $\textbf{m}^c$, covariance matrix $\textbf{S}^c$, and the coefficients of skewness $\textbf{g}^c=(g^c_i)_{i=1}^3$, of the random variable $  \textbf{X} \sim N(\textbf{m}, \textbf{S})$, where $ \textbf{m}=(1,1,1)$ and 
	$ \textbf{S} =
	\begin{bmatrix}
	2  &  1  &  1\\
	1  &  2  &  2\\
	1  &  2&  2
	\end{bmatrix}, $ when the censoring limits for the r.v. $X_3$ are $a =0.5$ and $b=2$. Then, we proceed as follows: 1) we produce $ 10^6$ random measurements from $ N(\textbf{m}, \textbf{S}) $ 100 times. 2) Each time, we calculate the sampling mean vector, the covariance matrix and the coefficients of skewness, derived from the censored measurements. 3) We calculate the average, $\textbf{m}_{sam}$, $ \textbf{S}_{sam} $ and $\textbf{g}_{sam}$ of the 100 samples mean vectors, covariance matrices and coefficients of skewness,  respectively. 4) The mean vector, $\textbf{m}^c$, the covariance matrix, $ \textbf{S}^c $ and the coefficients of skewness, $\textbf{g}^c$,  are calculated by (\ref{mrg1}), (\ref{cen_mean})-(\ref{cen_cov}). As it can be seen in (\ref{first})-(\ref{last}), the proposed $\textbf{m}^c$, $\textbf{S}^c$ and $\textbf{g}^c$, are almost identical with the corresponding results from the  sample, 
	
	\begin{align}
	\label{first}
	\textbf{m}_{sam} &= (0.9999, 1.0000, 1.1495)\\
	\textbf{m}^c &= (1.0000, 1.0000, 1.1494)\\
	\textbf{S}_{sam} &=
	\begin{bmatrix}
	2.0003  &  1.002  &  0.3985\\
	1.0002  &  1.9998  &  0.7968\\
	0.3985  &  0.7968&  0.4003
	\end{bmatrix}\\
	\textbf{S}^c &=
	\begin{bmatrix}
	2.0000  &  1.0000  &  0.3984\\
	1.0000  &  2.0000  &  0.7968\\
	0.3984  &  0.7968&  0.4003
	\end{bmatrix}\\
	\textbf{g}_{sam} &= (0.0001, -0.0001, 0.2654)\\
	\label{last} 
	\textbf{g}^c &= (0.0000, 0.0000, 0.2657).
	\end{align}

	\section{Tobit Kalman Filter with Coloured Noise}
	
	The state-space model with censored measurements and coloured noises is defined as
	
	\begin{align}
	\label{state_x}
	\textbf{x}_{t+1}&=\textbf{A}\textbf{x}_{t} +\textbf{u}_{t},\\
	\textbf{u}_{t}&=\textbf{C}\textbf{u}_{t-1} +\textbf{w}_{1,t},\\
	y_{t}&=
	\begin{cases}
	a, \quad y^*_{t} \leq a\\
	{y}^*_{t}=\textbf{H}\textbf{x}_{t} +{v}_{t},\quad a<y^*_{t}<b\\
	b, \quad y^*_{t} \geq b\\
	\end{cases}\\
	v_{t}&= g\cdot v_{t-1} + {w}_{2,t},
	\label{colored}
	\end{align}
	where $\textbf{A}$ and $\textbf{H}$ are the transition and observation matrix, respectively, $y^*_{t}$, $ y_t $ and $\textbf{x}_t\in\Re^n$ are the latent measurement, the censored measurement and the unknown state vector at time frame $t$, respectively, while $ a $ and $ b $ are the censoring limits.  $\textbf{w}_{1,t}$, ${w}_{2,t}$ are white noises (hence, of zero mean) with covariance matrix $\textbf{Q}$ and variance $r^2$, respectively, while $\textbf{u}_t$, $ v_t $ are coloured noises generated by the associated AR(1) models driven by  matrix $ \textbf{C} $ and scalar $ g $, respectively. To overcome the problem of the coloured noises, the system given by (\ref{state_x})-(\ref{colored}) is expressed  as a system with white noise using the augmented approach. For that purpose let $\textbf{A}_{aug} =     
	\begin{bmatrix}
	\textbf{A}  &  \textbf{I}_n & \textbf{0}\\
	\textbf{0}_n  &  \textbf{C} & \textbf{0}\\
	\textbf{0}^{T}  & \textbf{0}^{T} & g\\
	\end{bmatrix}$, 
	$\textbf{z}_t = 
	\begin{bmatrix}
	\textbf{x}_t \\
	\textbf{u}_t\\ 
	v_t\\
	\end{bmatrix}$, 
	$\textbf{H}_{aug} = [\textbf{H} \quad \textbf{0}^{'}\quad 1]$,
	$\textbf{w}_{aug,t} = 
	\begin{bmatrix}
	\textbf{0}\\
	\textbf{w}_{1,t}\\
	{w}_{2,t}\\ 
	\end{bmatrix}$ $\sim N(\textbf{0},\textbf{Q}_{aug}) $, 
	$\textbf{Q}_{aug} =     
	\begin{bmatrix}
	\textbf{0}_n  &  \textbf{0}_n & \textbf{0}\\
	\textbf{0}_n  &  \textbf{Q} & \textbf{0}\\
	\textbf{0}^{T}  & \textbf{0}^{T} & r^2\\
	\end{bmatrix}$, where $ \textbf{0}_n $ and $ \textbf{0} $ denote the $ n \times n $ zero matrix  and the $ n \times 1$ zero vector, respectively. 
	Then, the state-space model (\ref{state_x})-(\ref{colored}) can be written in the form: 
	
	\begin{align}
	\label{state_aug}
	\textbf{z}_{t+1}&=\textbf{A}_{aug}\textbf{z}_{t} +\textbf{w}_{aug,t+1},\\
	y_{t}&=
	\begin{cases}
	a, \quad y^*_{t} \leq a\\
	{y}^*_{t}=\textbf{H}_{aug}\textbf{z}_{t} \quad a<y^*_{t}<b \\
	b, \quad y^*_{t} \geq b
	\label{colored_aug}
	\end{cases}.
	\end{align}
	
	As it can be seen in (\ref{colored_aug}) the latent measurement $ y^*_t $ is noise-free (i.e., a perfect measurement). The linear optimal estimates for noise-free measurements (\ref{state_aug})-(\ref{colored_aug}) have the same form (with the corresponding new matrices in the augmented model) as in the case of the state-space model with white noises, except that the variance of measurement noise in the augmented model equals 0.
	
	In this paper, as in \cite{allik2015tobit,geng2017tobit,loumponias2018adaptive}, the a posteriori estimation of the state vector, $ \hat{\textbf{z}}_{t}$, is calculated as a linear combination of the a priori estimation of the state vector, $ \hat{\textbf{z}}^-_{t} $, and the censored measurement $ y_t $. Although these estimations are not optimal, it is proved that they minimize the trace of the state error covariance matrix \cite{anderson2012optimal}. More specifically, the proposed method -as in standard KF- evolves in two stages, the predict and the update stage, respectively:
	
	\textit{Predict Stage}
	\begin{align}
	\nonumber
	\hat{\textbf{z}}_t^- &=\E(\textbf{z}_t|y_{1:t-1}), \\ 
	\nonumber
	\textbf{P}_t^- &=Cov(\textbf{z}_t-\hat{\textbf{z}}_t^-|y_{1:t-1})
	\end{align}
	\textit{Update Stage}
	\begin{align}
	\nonumber
	\textbf{K}_t  &= Cov(\textbf{z}_t,y_t|y_{t-1})\cdot Var(y_t|y_{t-1})^{-1},\\
	\label{z_aug}
	\hat{\textbf{z}}_t  &= \hat{\textbf{z}}^-_t+\textbf{K}_t(y_t-\E(y_t|y_{t-1})),\\
	\label{P_aug}
	\textbf{P}_t  &= \textbf{P}^-_t-\textbf{K}_t\cdot Cov(\textbf{x}_t,y_t|y_{k-1})^{T},
	\end{align}
	where $\textbf{P}_t^-$ and $\textbf{P}_t$ are the covariance matrices of the a priori and a posteriori error estimations, respectively. $\E(y_t|y_{t-1})$ and $Var(y_t|y_{t-1})$ are the mean and variance of the censored measurement $y_t$ given the censored measurements up to time $ t-1$, while $Cov(\textbf{z}_t,y_t|y_{t-1})$ is the cross-covariance matrix of the augmented state and the censored measurement at time $t$.
	
	The predict stage is the same as in the  case of the standard KF, since the censored measurements are not used in this stage. Therefore the a priori estimations are given by 
	\begin{align}
	\hat{\textbf{z}}_t^- &= \textbf{A}_{aug}\hat{\textbf{z}}_{t-1},\\
	{\textbf{P}}_t^- &= \textbf{A}_{aug}{\textbf{P}}_{t-1}\textbf{A}_{aug}^T + \textbf{Q}_{aug}. 
	\label{P_prior}  
	\end{align}
	
	It is clear by (\ref{state_aug}) and (\ref{colored_aug}) that the joint distribution of $ \textbf{z}_t $ and $ y^*_t $ is Gaussian and more specifically $(\textbf{z}_t, y^*_t|t-1) \sim N(\textbf{m}, \textbf{S} )$, where
	$\textbf{m} = 
	\begin{bmatrix}
	\hat{\textbf{z}}_t^- \\
	\textbf{H}_{aug}\hat{\textbf{z}}_t^-\\ 
	\end{bmatrix}$ and 
	$\textbf{S} =     
	\begin{bmatrix}
	\textbf{P}_t^{-}  & Cov(\textbf{z}_t,y^*_t|t-1)  \\
	Cov(\textbf{z}_t,y^*_t|t-1) &  Var(y_t^{*}|t-1) \\
	\end{bmatrix}$. Furthermore, it can be proven that
	$ Cov(\textbf{z}_t,y^*_t|t-1) = \textbf{P}^-_{t}\textbf{H}_{aug}^{T} $ and $ Var(y_t^{*}|t-1) = \textbf{H}_{aug}\textbf{P}^-_{t} \textbf{H}_{aug}^{T}$. Next, the results of Proposition \ref{p1} are utilized to cope with the censored moments in the Update Stage.  The censored moments $\E(y_t|t-1)$, $Var(y_t|t-1)$ and $Cov(\textbf{z}_t,y_t|t-1)$ are calculated by (\ref{cen_mean}), (\ref{cen_var}) and (\ref{cen_cov}), respectively. Summarizing, the proposed Update Stage for censored measurements is calculated as follows:
	
	\begin{enumerate}
		
		\item $\E(y_t|t-1)$:   is calculated by (\ref{cen_mean}), by substituting $m_k = \textbf{H}_{aug}\hat{\textbf{z}}_t^{-} $ and $S_{k,k} =  \textbf{H}_{aug}\textbf{P}^-_{t} \textbf{H}_{aug}^{T}$.
		
		\item  $Var(y_t|t-1)$:   is calculated by (\ref{cen_var}).
		
		\item $Cov(\textbf{z}_t,y_t|t-1)$:   is calculated by (\ref{cen_cov}) and is equal with,
		\begin{equation}
		Cov(\textbf{z}_t,y_t|t-1) =  \textbf{P}^-_{t}\textbf{H}_{aug}^{T}\cdot P,
		\end{equation}
		where $P$ is the probability of $y^*_t$ to belong into the uncensored interval $(a,b)$ and is given by, 
		\begin{equation*}
		P=F_{e_k}(b-m_k)-F_{e_k}(a-m_k),
		\end{equation*}
		where $e_k\sim N(0,S_{k,k})$.
		
	\end{enumerate}
	Then, the estimation of the state vector $\textbf{z}_t$ and the corresponding covariance matrix of the estimation error are calculated by (\ref{z_aug}) and (\ref{P_aug}), respectively.
	Hereinafter, we denote the filtering process that takes into account the corrected censored moments and does not consider coloured noises,  by TKF$^c$.
	
	The proposed process described by (\ref{z_aug})-(\ref{P_prior}) can only be applied when the AR(1) models of coloured noises ($\textbf{C}$ and $ g $) are assumed to be known.  However, in real-life problems these parameters are unknown. In order to overcome this problem, the Likelihood Function (LF) of the censored measurements is utilised to estimate the parameters of the coloured noises.
	The LF for the censored measurements $ \{y_{t}\}_{t=1}^{T} $ given in (\ref{colored_aug}) with censoring limits $ a $ and $ b $ has the form \cite{loumponias2016using,loumponias2016usingb},
	
	\begin{align}
	L(\textbf{y}) =& {\displaystyle\prod_{ a <~ y_{t} <~b} \frac{1}{(\textbf{H}_{aug}\textbf{P}^-_{t}\textbf{H}_{aug}^T)^{1/2}} \phi    \Bigg(\frac{y_{t}-\textbf{H}_{aug}\hat{\textbf{z}}^-_{t}}{(\textbf{H}_{aug}\textbf{P}^-_{t}\textbf{H}_{aug})^{1/2}} \Bigg) } \nonumber\\  
	&\times{\displaystyle\prod_{y_{t}=a}\Phi\Bigg(\frac{a-\textbf{H}_{aug}\hat{\textbf{z}}^-_{t}}{(\textbf{H}_{aug}\textbf{P}^-_{t}\textbf{H}_{aug})^{1/2}}\Bigg)}  \nonumber\\ 
	&\times{\displaystyle\prod_{y_{t}=b}\Bigg(1-\Phi\Bigg(\frac{b-\textbf{H}_{aug}\hat{\textbf{z}}^-_{t}}{(\textbf{H}_{aug}\textbf{P}^-_{t}\textbf{H}_{aug})^{1/2}}\Bigg)}
	\label{likehood}\Bigg).
	\end{align}

	\section{Experiments}

	In this section, two experiments-simulations are conducted to evaluate ColTKF in comparison to the standard augmented  KF (AKF) and TKF$^c$. More specifically, two oscillators (without damping) are considered, which have been utilised frequently in literature  \cite{allik2015tobit},\cite{geng2018distributed},\cite{geng2017tobit},\cite{han2018improved}. In the first experiment, the noises of the state-space model are coloured, while, in the second, they are white. 
	
	Let the state space equations have the form of (\ref{state_x})-(\ref{colored}) with
	$ 	\textbf{H}=
	\begin{bmatrix}
	1&0.5
	\end{bmatrix},$
	$ 	\textbf{A}=
	\begin{bmatrix}
	cos(\omega)&-sin(\omega)\\
	sin(\omega)&\quad cos(\omega)
	\end{bmatrix},$	
	where $ \omega=0.005\cdot2\pi$. The disturbances $\textbf{w}_{1,t} $ and $ w_{2,t} $ are assumed to be normally distributed, i.e., $\textbf{w}_{1,t}\sim N(\textbf{0},\textbf{Q}) $ and $ w_{2,t} \sim N(0, r^2) $, where
	$ 	\textbf{Q}=
	\begin{bmatrix}
	0.01^2&0\\
	0&0.01^2
	\end{bmatrix}$
	and $r^2 = 1$.  In the first experiment, the coloured noise parameters are set as 
	$ \textbf{C}_1 = 
	\begin{bmatrix}
	0.9&0\\
	0&0.9
	\end{bmatrix}
	$
	and $ g_1 = 0.99,$ while in the second experiment, they are equal with $ \textbf{C}_2 = \textbf{0}_2 $ and $ g_2 = 0$, i.e., coloured noises are not considered. Moreover, in the first experiment, the censoring limits are equal with $ a_1 = -5 $ and $ b_1 = 5 $, while in the second experiment, they are equal with  $ a_2 = -1 $ and $ b_2 = 1 $.

	Let the initial state vector be  $ \textbf{x}_0=[5\quad 0]^T $ with covariance matrix $ \textbf{P}_0 = 10^{-3}\cdot\textbf{I}_{2 \times 2} $. Then, by the above parameters, censored (saturated) measurements, $y_t$, are produced for $t=1,2,...,500$. In order the results of the experiments to be more valuable,  the above process is repeated 100 times (Monte Carlo simulations) and in each Monte Carlo simulation the root mean square errors (RMSEs) of the three methods (i.e., AKF, 
	TKF$^c$ and ColTKF) are calculated. Moreover, only in the case of ColTKF, the parameters of AR(1) models, i.e, $\{ \textbf{C}_1,\, g_1\} $ and  $\{ \textbf{C}_2,\, g_2\} $ are assumed to be unknown for both experiments; thus, LF (\ref{likehood}) is used (in each simulation) in order to estimate them. 
	
	The means of the filters’ RMSEs for the first experiment (for the 100 simulations) are presented in Table \ref{rmse1}, where the means of RMSEs for both coordinates of the state vector $\textbf{x}_t$ are provided. It is clear by Table \ref{rmse1} that the AKF process has a poor performance, since it is not able to deal with censored measurements (see Fig. \ref{fig1}). TKF$^c$ has a better performance than AKF, since it considers the censoring measurements, but, it can not cope with the coloured noises. Finally, the proposed method ( ColTKF) has the best performance overall, since it takes into account: a) the censoring limits in the measurements by calculating the accurate censored moments (Section 2) and b) the heteroskedasticity by estimating the parameters $\{ \textbf{C}_1,\, g_1\} $ via LF (\ref{likehood}). In Fig. \ref{figx1} the methods' estimations for the hidden states vector $\textbf{x}_t$ (yellow plot) are illustrated . It is clear, that AKF and TKF$^c$ provide biased estimations  due to the censoring and the coloured noises, while ColTKF tackles both problems.

	\begin{table}[h!]
		\renewcommand{\arraystretch}{1.3}
		\begin{center}
			\begin{tabular}{ |c|c|c| }
				\hline
				\textbf{Filter}  & \textbf{Mean RMSE of} $\hat{\textbf{x}}_1 $ &  \textbf{Mean RMSE of} $\hat{\textbf{x}}_2 $ \\
				\hline
				AKF  & 10.1292 & 10.4497 \\
				TKF$^c$ & 8.7346 & 9.0072 \\
				\textbf{ColTKF} & \textbf{6.2879} & \textbf{6.9183} \\
				\hline
			\end{tabular}
		\end{center}
		\caption{The means of the RMSEs for the filters AKF, TKF$^c$ and ColTKF, respectively, for the first experiment.}
		\label{rmse1}
	\end{table} 
	
	\begin{figure}[h!]
		\centering
		\includegraphics[width=12cm]{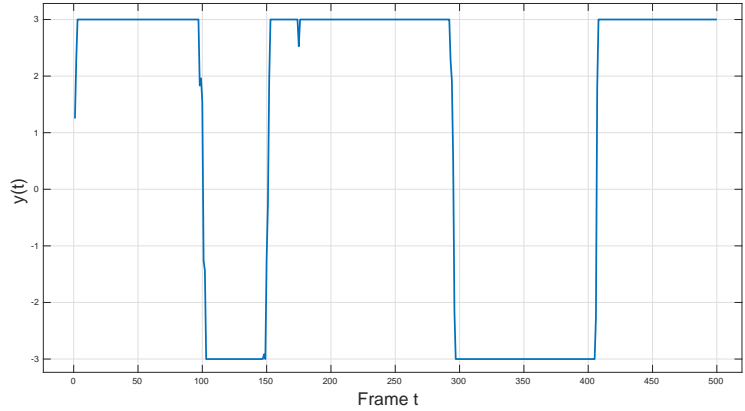}
		\caption{Censored measurements of the first experiment.}
		\label{fig1}
	\end{figure}

	\begin{figure}[h!]
		\centering
		\includegraphics[width=12cm]{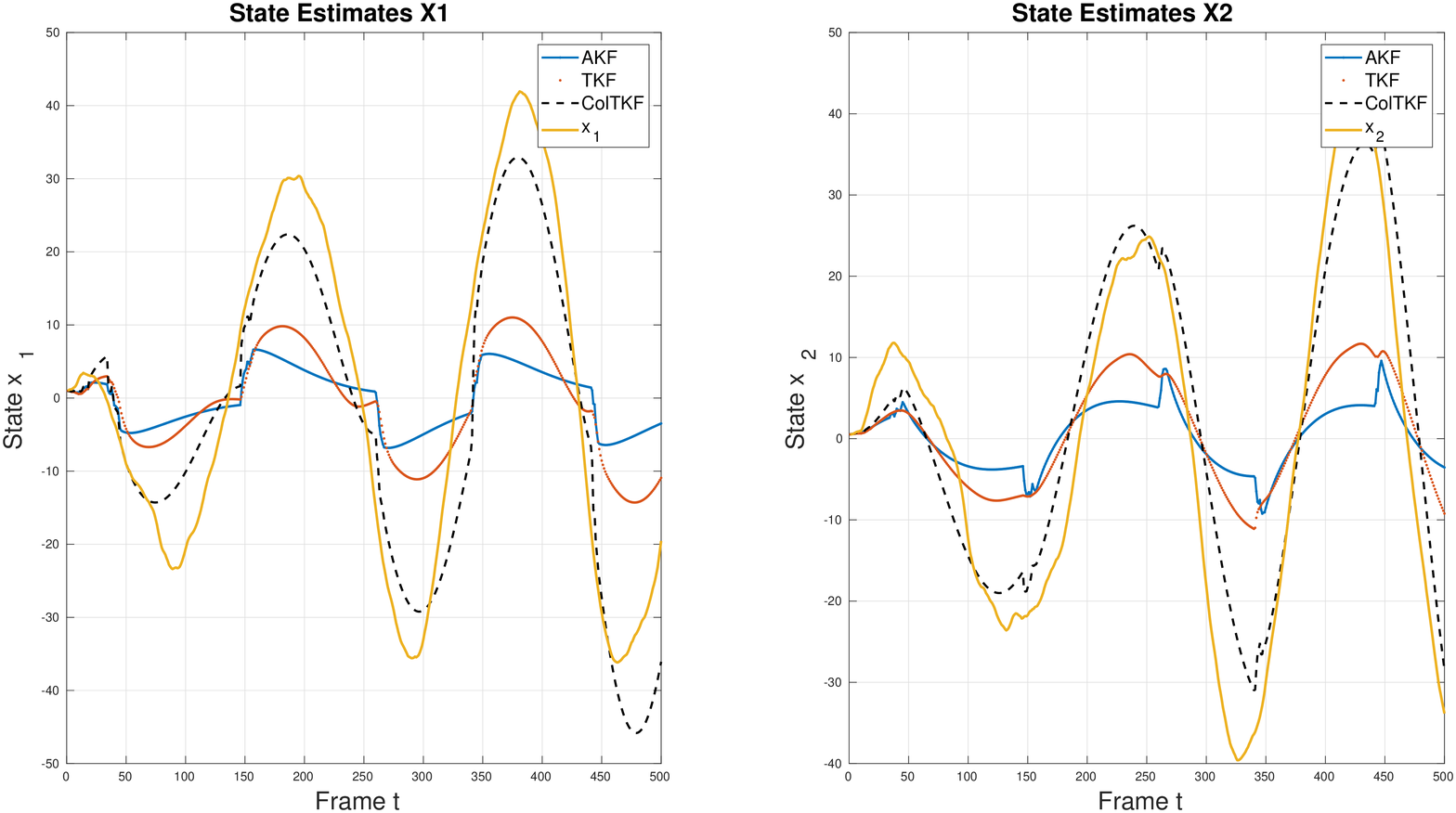}
		\caption{State estimates fo the filters AKF , TKF$^c$ and ColTKF, respectively, for the first experiment.}
		\label{figx1}
	\end{figure}
	
	In the same way as in the first experiment, the means of the filters’ RMSEs for the second experiment are presented in Table \ref{rmse2}. In the second experiment, the AKF process coincides with the corresponding one of the standard KF, since the system noises are not coloured. For the same reason, TKF$^c$ process coincides with the proposed ColTKF. It is clear by Table \ref{rmse2} that the AKF process has a poor performance, since it is not able to deal with censored measurements (see Fig. \ref{fig2}). TKF$^c$ has almost the same performance (not exactly the same) as ColTKF, since, in the proposed method the parameters $\{ \textbf{C}_2,\, g_2\} $ are estimated by LF (\ref{likehood}) and they are not assumed to be known. In Fig. \ref{figx2}  the methods' estimations for the hidden states vector $\textbf{x}_t$ (yellow plot) are illustrated. As it can be seen, TKF$^c$ and ColTKF have the same performance, while AKF provides biased estimates when the measurements belong into the censored region.           
	
	\begin{table}[h!]
		\renewcommand{\arraystretch}{1.3}
		\begin{center}
			\begin{tabular}{ |c|c|c| }
				\hline
				\textbf{Filter}  & \textbf{Mean RMSE of} $\hat{\textbf{x}}_1 $ &  \textbf{Mean RMSE of} $\hat{\textbf{x}}_2 $ \\
				\hline
				AKF & 0.5784 & 0.6218 \\
				TKF$^c$ & 0.4691 & 0.5163 \\
				\textbf{ColTKF} & \textbf{0.4671} & \textbf{0.5156} \\
				\hline
			\end{tabular}
		\end{center}
		\caption{The means of the RMSEs for the filters AKF, TKF$^c$ and ColTKF, respectively, for the second experiment.}
		\label{rmse2}
	\end{table} 
	
	\begin{figure}[h!]
		\centering
		\includegraphics[width=12cm]{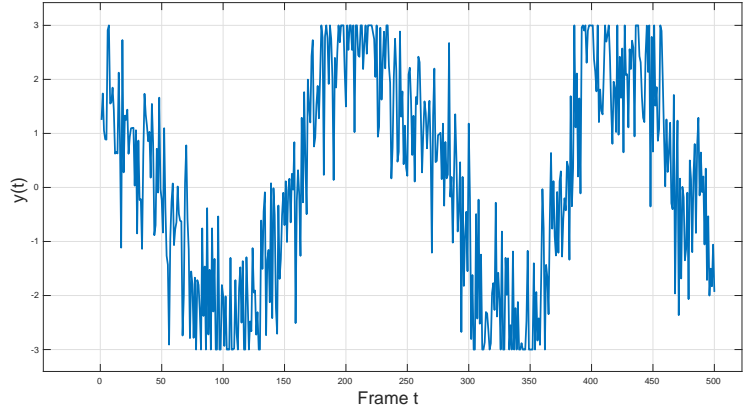}
		\caption{Censored measurements  of the second experiment.}
		\label{fig2}
	\end{figure}
	
	\begin{figure}[h!]
		\centering
		\includegraphics[width=12cm]{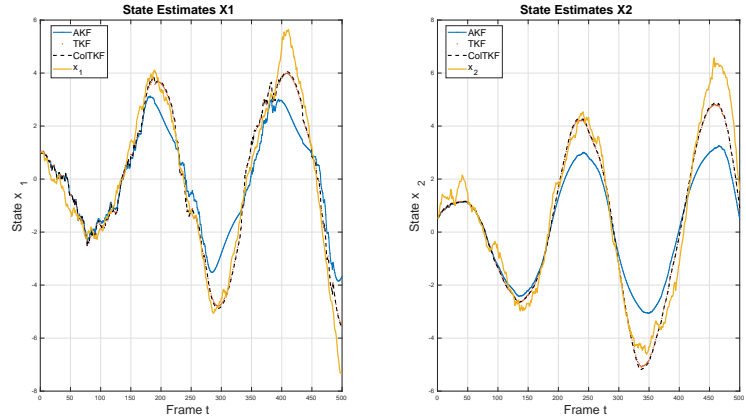}
		\caption{State estimates fo the filters $ AKF $, $TKF^c$ and $ ColTKF $, respectively,  for the second experiment.}
		\label{figx2}
	\end{figure}

	\section{Conclusion}
	
	The aim of this paper is to improve the TKF process by 1) calculating the exact censored moments and 2) by considering coloured noises for the state-space model. To that end, the mgf of a censored normal distribution with two censoring limits was calculated. Then, the exact censored moments were calculated by utilizing the associated mgf. Next, in the proposed method, the augmented approach was used in order to deal with coloured noises which are described by AR(1) models. Furthermore, LF of the censored measurements was provided in order to estimate the unknown parameters of the AR(1) models.  
	
	The proposed method, ColTKF, was evaluated against TKF$^c$ and AKF in two different simulations-experiments. In the first experiment, the state-space model describes the motion of an oscillator, where system's noises are assumed to be coloured, while in the second experiments the noises are assumed to be white. In the proposed method, the AR(1) parameters were set to be unknown, thus, LF was utilised in advance to estimate them. It is worth to mention that in each Monte Carlo simulation, only the estimated parameters of $\{\textbf{C}_1,\,g_1\}$ and $\{\textbf{C}_2,\,g_2\}$ were utilised in  ColTKF. In the first experiment, the results showed that ColTKF outperforms (minimum RMSE) both TKF$^c$ and AKF. This result was expected, since AKF cannot handle the censored measurements and TKF$^c$ cannot handle   coloured noises. In the second experiment,  ColTKF and TKF$^c$ appear to have almost the same performance, since the noises are white, while AKF provides biased estimations, when the measurements are censored. Therefore, the proposed method, ColTKF, in both experiments is able to detect whether coloured noises are present in the system (first experiment) or not (second experiment) and then to estimate the hidden state vectors $\{\textbf{x}_t\}_{t=1}^T$.   
	Moreover, as a step further it would be interesting to extend the proposed method in multidimensional censored measurements with correlated coordinates, to
	describe efficiently real-life problems with censored measurements, when the noises in the state-space model are coloured.
	
	

\end{document}